\documentclass[11pt,letterpaper]{amsart}

\usepackage{amsmath,amsfonts,amsthm,amssymb,stmaryrd,bm,cite,enumerate}
\usepackage{relsize,tabls}

\theoremstyle{plain}

\numberwithin{equation}{section}

\newtheorem{thm}{Theorem}[section]
\newtheorem{lem}[thm]{Lemma}

\theoremstyle{definition}
\newtheorem{example}{Example}

\theoremstyle{definition}
\newtheorem{problem}{Problem}

\allowdisplaybreaks  % introduced 7.15.15

\newcommand{\complex}{{\mathbb C}}
\newcommand{\Natural}{{\mathbb N}}
\newcommand{\real}{{\mathbb R}}

\newcommand{\rmcyl}{\mathrm{cyl\,}}   % added 11.1.23

\newcommand{\ascript}{\mathcal{A}}
\newcommand{\cscript}{\mathcal{C}}

\newcommand{\muhat}{\widehat{\mu}}

\newcommand{\ab}[1]{\left|#1\right|}
\newcommand{\brac}[1]{\left\{#1\right\}}
\newcommand{\paren}[1]{\left(#1\right)}
\newcommand{\sqbrac}[1]{\left[#1\right]}
\newcommand{\sqparen}[1]{{\left[#1\right)}} % \left [ and \right ) - in search will look for closing ] or (

\begin{document}

% 11.1.23 remove \\
\title{FOUR RELATED COMBINATORIAL PROBLEMS}
% 11.1.23 add \address, \email
\author{Stan Gudder}
\address{Department of Mathematics, 
University of Denver, Denver, Colorado 80208}
\email{sgudder@du.edu}
\date{}
\maketitle

% add 11.12.23 for arXiv submission
\begin{abstract}
This pedagogical article solves an interesting problem in quantum measure theory. Although a quantum measure $\mu$ is a generalization of an ordinary probability measure, $\mu$ need not satisfy the usual additivity condition. Instead, $\mu$ satisfies a grade-2 additivity condition. Besides the quantum measure problem, we present three additional combinatorial problems. These are 
(1)\enspace A sum of binomial coefficients problem;\enspace
(2)\enspace A recurrence relation problem; and\enspace
(3)\enspace An interated vector problem. We show that these three problems are equivalent in that they have a common solution. We then show that this solves the original quantum measure problem.
\end{abstract}

\section{Introduction}  % Section 1
An elegant combinatorial proof for the sums of evenly spaced binomial coefficients is given in \cite{bck10}. The present article shows that a special case of this result can be used to solve an interesting problem in quantum measure theory. In more detail we present four combinatorial problems which at first sight appear to be unrelated. We then show that the last three problems are equivalent; that is, they have the same solution, and the solution to the first problem follows from the other three. Although the first problem from quantum measure theory appears to be quite difficult, viewing it in one of the three other ways suggests methods for its solution. We then present three quite different ways to solve the problems all of which differ from the method in \cite{bck10}

\section{A Quantum Measure Problem}  % Section 2
This problem from quantum measure theory \cite{gou72,gud101,gui95,lov07} is the first of these four problems that the author attempted to solve. Suppose we have a particle that can occupy one of two sites which we call 0 and 1. The particle begins at site 0 at time $t=0$ and then evolves in discrete time steps $t=1,2,\ldots\,$. We do not know the exact motion of this particle and all we know is the probability that the particle remains at its present site or moves to the other site in one time step.

If this were a classical particle we would have a 2-site random walk which is a process governed by a stochastic matrix $T=\sqbrac{t_{jk}}$, $j,k=0,1$, where $t_{jk}$ is the transition probability that the particle moves from site $k$ to site $j$ in one time step. Hence, $0\le t_{jk}\le 1$ and 
\begin{equation*}
t_{00}+t_{10}=t_{01}+t_{11}=1
\end{equation*}
Since the particle begins at site 0, the initial distribution is given by the vector $\left[\begin{smallmatrix}\smallskip 1\\0\end{smallmatrix}\right]$.
At time $t=1$, the distribution becomes
\begin{equation*}
T\begin{bmatrix}1\\0\end{bmatrix}=\begin{bmatrix}t_{00}&t_{01}\\t_{10}&t_{11}\end{bmatrix}\begin{bmatrix}1\\0\end{bmatrix}
  =\begin{bmatrix}t_{00}\\t_{10}\end{bmatrix}
\end{equation*}
Thus, at $t=1$ the particle is at 0 with probability $t_{00}$ and at 1 with probability $t_{10}$. At time $t=n$ the distribution of the particle is
$T^n\left[\begin{smallmatrix}\smallskip 1\\0\end{smallmatrix}\right]$. Now consider a path $\gamma _0\gamma _1\cdots\gamma _n$,
$\gamma _0=0$, $\gamma _i\in\brac{0,1}$, $i=1,2,\ldots ,n$. The probability that the particle moves along this path is
\begin{equation*}
t_{\gamma _1\gamma _0}t_{\gamma _2\gamma _1}\cdots t_{\gamma _n\gamma _{n-1}}
\end{equation*}
For example at $t=2$ there are four possible paths: $000$, $001$, $010$ and $011$. The probabilities for these paths are $t_{00}t_{00}$,
$t_{00}t_{10}$, $t_{10}t_{01}$ and $t_{10}t_{11}$, respectively. Of course, the sum of these probabilities is
\begin{equation*}
t_{00}(t_{00}+t_{10})+t_{10}(t_{01}+t_{11})=t_{00}+t_{10}=1
\end{equation*}

But the particle under consideration is not a classical particle, it is a quantum particle so we have a 2-site quantum random walk. In this case, we cannot use the classical theory of stochastic processes but must employ the formalism of quantum mechanics. The initial distribution is still
$\left[\begin{smallmatrix}\smallskip 1\\0\end{smallmatrix}\right]$ but now $\left[\begin{smallmatrix}\smallskip 1\\0\end{smallmatrix}\right]$ is considered to be a complex vector in $\complex ^2$ and $T$ is replaced by a unitary matrix $U=\sqbrac{u_{jk}}$ on $\complex ^2$. At time $t=n$, the vector
\begin{equation*}
\begin{bmatrix}v_0\\v_1\end{bmatrix}=U^n\begin{bmatrix}1\\0\end{bmatrix}
\end{equation*}
now gives the ``amplitude'' distribution so that $\ab{v_0}^2$ and $\ab{v_1}^2$ are the probabilities the particle is at 0 and 1, respectively,
at time $n$. The \textit{amplitude} of the path $\gamma=\gamma _0\gamma _1\ldots\gamma _n$ becomes
\begin{equation}                % equation (2.1)
\label{eq21}
a(\gamma )=u_{\gamma _1\gamma _0}u_{\gamma _2\gamma _1}\cdots u_{\gamma _n\gamma _{n-1}}
\end{equation}
and the probability the particle moves along $\gamma$ is $\ab{a(\gamma )}^2$.

We shall assume that $U$ has the simplest nontrivial form
\begin{equation*}
U=\tfrac{1}{\sqrt{2}}\begin{bmatrix}1&i\\i&1\end{bmatrix},\quad i=\sqrt{-1}
\end{equation*}
In this case, the particle remains at its present site for one time step with amplitude $1/\sqrt{2}$ and moves to the other site with amplitude $i\sqrt{2}$.
The evolution of the particle is described by an infinite bit string $\gamma =\gamma _0\gamma _1\cdots$, $\gamma _0=0$, $\gamma _j\in\brac{0,1}$, $j=1,2,\ldots$, which specifies that the particle is at site $\gamma _j$ at time $j$. We call $\gamma$ a \textit{path} or \textit{trajectory} of the particle. We also consider
$n$-\textit{paths} $\gamma =\gamma _0\gamma _1\cdots\gamma _n$ which specify sites occupied by the particle at times $0,1,\ldots ,n$. The general
\textit{sample space} is the set
\begin{equation*}
\Omega =\brac{\gamma\colon\gamma\hbox{ a path}}
\end{equation*}
and the \textit{time}-$n$ \textit{sample space} is the set
\begin{equation*}
\Omega _n=\brac{\gamma\colon\gamma\hbox{ an }n\hbox{-path}}
\end{equation*}
Notice that the cardinality $\ab{\Omega _n}=2^n$.

By \eqref{eq21} the amplitude of an $n$-path $\gamma =\gamma _0\gamma _1\cdots\gamma _n$ is
\begin{equation}                % equation (2.2)
\label{eq22}
a(\gamma )=\tfrac{1}{2^{n/2}}\,i^{\ab{\gamma _1-\gamma _0}}i^{\ab{\gamma _2-\gamma _1}}\cdots i^{\ab{\gamma _n-\gamma _{n-1}}}
\end{equation}
For example, for the 4-path $\gamma =01101$ we have $a(\gamma )=-i/4$. If $\gamma =\gamma _0\gamma _1\cdots\gamma _n$ and
$\gamma '=\gamma '_0\gamma '_1\cdots\gamma '_n$ are $n$-paths, their \textit{decoherence} is defined by
\begin{equation}                % equation (2.3)
\label{eq23}
D^n(\gamma ,\gamma ')=a(\gamma )\overline{a(\gamma ')}\delta _{\gamma _n,\gamma '_n}
\end{equation}
Here, $\delta _{j_k}$ is the Knonecker delta so $D^n(\gamma ,\gamma ')=0$ unless $\gamma and \gamma '$ end at the same site. Denoting the collection of all subsets of $\Omega _n$ by $\ascript _n$, we call the elements of $\ascript _n$ \textit{time}-$n$ \textit{events} and define the \textit{quantum measure}
$\mu _n\colon\ascript _n\to\real$ by
\begin{equation*}
\mu _n(A)=\sum\brac{D^n(\gamma ,\gamma ')\colon\gamma ,\gamma '\in A}
\end{equation*}

\begin{example}  % Example 1
For $n=2$, let $\omega _0=000$, $\omega _1=001$, $\omega _2=010$, $\omega _3=011$ so that
$\Omega _2=\brac{\omega _0,\omega _1,\omega _2,\omega _3}$. The reader can now verify that $\mu _2\paren{{\emptyset}}=0$, $\mu _2(\brac{\omega _j})=1/4$,
$j=0,1,2,3$, $\mu _2\paren{\brac{\omega _0,\omega _2}}=0$\newline
$\mu _2\paren{\brac{\omega _0,\omega _1}}=\mu _2\paren{\brac{\omega _0,\omega _3}}=\mu_2\paren{\brac{\omega _1,\omega _2}}
   =\mu _2\paren{\brac{\omega _2,\omega _3}}=1/2$\newline
$\mu _2\paren{\brac{\omega _1,\omega _3}}=1$, $\mu _2\paren{\brac{\omega _0,\omega _1,\omega _2}}=1/4$, $\mu _2(\Omega _2)=1$ and\newline
$\mu _2\paren{\brac{\omega _0,\omega _1,\omega _3}}=\mu _2\paren{\brac{\omega _1,\omega _2,\omega _3}}=5/4$\hfill\qedsymbol
\end{example}

% typed 10.31.23 in 4relatcombProblems continued.tex - copied here - BEGIN
Readers can try their hand at showing that $\mu _n$ satisfies:
\begin{enumerate}[(a)]
%(a)
\item $\mu _n(A)\ge 0$ for all $A\in\ascript _n$.
%(b)
\item $\mu _n(\Omega _n)=1$.
\end{enumerate}
We will indicate later why (a) and (b) hold. Although $\mu _n$ is sometimes thought of as a quantum probability, Example~1 shows that $\mu _n$ does not satisfy some of the properties of a probability measure. For example, $A\subseteq B$ does not imply that $\mu _n(A)\le\mu _n(B)$ and
$\mu _n(A\cup B)\ne\mu _n(A)+\mu _n(B)$ whenever $A\cap B=\emptyset$. However, it can be shown that $\mu _n$ satisfies the following
\textit{grade-2 additivity condition} \cite{gud101,gud102,sor94,sor07}.
  
\begin{enumerate}[(c)]
\item If $A,B,C\in\ascript _n$ are mutually disjoint, then
\end{enumerate}
\begin{equation*}
\mu _n(A\cup B\cup C)=\mu _n(A\cup B)+\mu _n(A\cup C)+\mu _n(B\cup C)-\mu _n(A)-\mu _n(B)-\mu _n(C)
\end{equation*}
Even though $\mu _n$ is not a traditional probability measure, $\mu_n(A)$ is thought of as the propensity that the event $A$ occurs.

We have defined $\mu _n\colon\ascript _n\to\real$ but we are also interested in the quantum measure of events $A\subseteq\Omega$. Can we extend $\mu _n$ to such sets? As a first step, we define a set $A\subseteq\Omega$ to be a \textit{cylinder set} if $A$ has the form
\begin{equation}                % equation (2.4)
\label{eq24}
A=A_1\times\brac{0,1}\times\brac{0,1}\times\cdots
\end{equation}
where $A_1\in\ascript _n$ for some $n\in\Natural$. For example, if $\gamma '=\gamma '_0\gamma '_1\cdots\gamma '_n\in\Omega _n$ we define the cylinder set
\begin{equation*}
\rmcyl (\gamma ')=\brac{\gamma '}\times \brac{0,1}\times\brac{0,1}\times\cdots =\brac{\gamma\in\Omega\colon\gamma _j=\gamma '_j, j=0,1,\ldots ,n}
\end{equation*}
% 11.5.23 finish above equation as I continue typing
% typed 10.31.23 in 4relatcombProblems continued.tex - END
We call $\rmcyl (\gamma ')$ an \textit{elementary cylinder set} and it is clear that every cylinder set is a disjoint union of a finite number of elementary cylinder sets. We denote the collection of cylinder sets by $\cscript (\Omega )$ and $\cscript (\Omega )$ is closed under complements and finite unions and intersections. For a cylinder set of the form \eqref{eq24}, we define $\mu (A)=u _n(A_1)$. It is not hard to show that $\mu\colon\cscript (\Omega )\to\real$ is well-defined and has the properties (a), (b), (c) of a quantum measure.

We now come to a deep and important problem of quantum measure theory. This problem is at least 60 years old and goes back to work on Feynman integrals
\cite{bur07,djs10,sor94,sor07}. Can we extend $\mu\colon\cscript (\Omega )\to\real$ to other physically important subsets of $\Omega$ in a meaningful and systematic way? For example, suppose $A\subseteq\Omega$ and we can identify a natural decreasing sequence $A_n\in\cscript (\Omega )$ such that $A=\cap A_n$. If
$\lim\mu (A_n)$ exists, then we can define $\muhat (A)=\lim\mu (A_n)$ and $\muhat$ would give a meaningful extension of $\mu$. Unfortunately, there are examples of decreasing sequences $A_n\in\cscript (\Omega )$ for which $\lim\mu (A_n)$ does not exist and we can even have $\lim\mu (A_n)=\infty$ \cite{djs10}. Of course, we are considering a very simple case. Of more physical interest would be an infinite number of sites (say all of $\real ^3$) and a continuum of times $t\in\sqparen{0,\infty}$. However, this simple case still illustrates the difficulties involved.

\begin{example}  % Example 2
For $\gamma\in\Omega$, the set $\brac{\gamma}\subseteq\Omega$ is certainly physically important and $\brac{\gamma}\notin\cscript (\Omega )$. Let
$\gamma =\gamma _0\gamma _1\cdots$ and form the sets $A_n=\rmcyl (\gamma _0\gamma _1\cdots\gamma _n)$. Then $A_n$ is a decreasing sequence of sets in $\cscript (\Omega )$ such that $\brac{\gamma}=\cap A_n$. The set $A_n$ can be naturally viewed as the time-$n$ approximation to ${\gamma}$. As in Example~1, it is easy to show that $\mu _n\paren{\brac{\omega}}=1/2^n$ for every $\omega\in\Omega _n$. Hence,
\begin{equation*} 
\lim _{n\to\infty}\mu (A_n)=\lim _{n\to\infty}\mu _n\paren{\brac{\gamma _0\gamma _1\cdots\gamma _n}}=\lim _{n\to\infty}\tfrac{1}{2^n}=0
\end{equation*}
so we can define $\muhat\paren{\brac{\gamma}}=0$ in a meaningful way. A similar analysis gives that $\muhat (A)=0$ for any finite set $A\subseteq\Omega$.
\hfill\qedsymbol
\end{example}

The next logical step after Example~2 is to consider the complement $\brac{\gamma}'$ of $\brac{\gamma}$, As before, let $\gamma =\gamma _0\gamma _1\cdots$ and let $A_n=\rmcyl (\gamma _0\gamma _1\cdots\gamma _n)$. Then $A'_n$ is an increasing sequence in $\cscript (\Omega )$ and $\brac{\gamma}'=\cup A'_n$. Although our result will hold for every $\gamma\in\Omega$, for simplicity let us assume that $\gamma _i=0$, $i=0,1,\ldots\,$. The set $\brac{\gamma}'$ represents the important physical event that the particle ever reaches site~1. In a similar way, $A=\brac{0111\cdots}'$ represents the event that the particle ever returns to site~0 and $\mu (A)$, if it is defined, would give the propensity (probability) of return to the origin. We now want to find $\lim\mu (A'_n)$ provided it exists. To this end, and it is also of interest in its own right, we would like to find $\mu (A'_n)$. Unlike the simple computation $\mu (A_n)=1/2^n$, finding $\mu (A'_n)$ is much more difficult. For instance $\mu (A'_1)=1/2$ and Example~1 shows that 
\begin{equation*}
\mu (A'_2)=\mu\paren{\brac{\omega _1,\omega _2,\omega _3}}=5/4
\end{equation*}
A little work gives $\mu (A'_3)=13/8$, $\mu (A'_4)=25/16$ and as we shall see $\mu (A'_n)$ exhibits an oscillatory behavior for increasing $n$. To alleviate the reader's suspense, we shall eventually show the expected result $\lim\mu (A'_n)=1$. We thus arrive at our first problem.

\begin{problem}   % Problem 1
Find $\mu (A'_n)$, $n=1,2,\ldots\,$.
\end{problem}

\section{Sums of Binomial Coefficients}  % Section 3
It is well-known that
\begin{equation}                % equation (3.1)
\label{eq31}
\sum _{j=0}^n\binom{n}{j}=2^n
\end{equation}
The easiest way to prove \eqref{eq31} is to let $a=b=1$ in the binomial theorem
\begin{equation}                % equation (3.2)
\label{eq32}
(a+b)^n=\sum _{j=0}^n\binom{n}{j}a^jb^{n-j}
\end{equation}
Another way to prove \eqref{eq31} is to note that the right side of \eqref{eq31} is the total number of subsets of a set $S$ of cardinality $n$ and that
$\binom{n}{j}$ is the number of $j$ element subsets of $S$. It is less well-known that if $n$ is even then
\begin{equation}                % equation (3.3)
\label{eq33}
\sum _{j=0}^{n/2}\binom{n}{2j}=\sum _{j=0}^{n/2-1}\binom{n}{2j+1}=2 ^{n-1}
\end{equation}
and a similar result holds if $n$ is odd. We can prove \eqref{eq33} by letting $a=-1$, $b=1$ in \eqref{eq32} to obtain
\begin{equation}                % equation (3.4)
\label{eq34}
\sum _{j=0}^n(-1)^j\binom{n}{j}=0
\end{equation}
Adding \eqref{eq34} to \eqref{eq31} and subtracting \eqref{eq34} from \eqref{eq31} gives \eqref{eq33}. We can combine the even and odd cases to obtain the following formula that holds for all $n\in\Natural$. Letting $\lfloor{x}\rfloor$ be the floor function which is the largest integer less than or equal to $x$ we have
\begin{equation}                % equation (3.5)
\label{eq35}
\sum _{j=0}^{\lfloor\frac{n}{2}\rfloor}\binom{n}{2j}=\sum _{j=0}^{\lfloor\frac{n-1}{2}\rfloor}\binom{n}{2j+1}=2^{n-1}
\end{equation}

Equation \eqref{eq35} gives sums of every second binomial coefficient. What about sums of every fourth binomial coefficient? For example,
\break
\begin{align}                % equation (3.6)
\label{eq36}
\binom{12}{0}&+\binom{12}{4}+\binom{12}{8}+\binom{12}{12}=992\notag\\
\binom{12}{1}&+\binom{12}{5}+\binom{12}{9}=1024\\
\binom{12}{2}&+\binom{12}{6}+\binom{12}{10}=1056\notag\\
\binom{12}{3}&+\binom{12}{7}+\binom{12}{11}=1024\notag
\end{align}
A closed form formula for sums of every fourth binomial coefficient is not well-known. Formulas for sums of evenly spaced binomial coefficients are given in
\cite{gou72,gui95,lov07} and most recently a combinatorial proof appeared in \cite{bck10}.

We then arrive at our second problem.

\begin{problem}   % Problem 2
For $n\in\Natural$ and $j=0,1,2,3$ find a closed form solution of
\begin{equation}                % equation (3.7)
\label{eq37}
\binom{n}{j}+\binom{n}{j+4}+\binom{n}{j+8}+\cdots +\binom{n}{4\lfloor\frac{n-j}{4}\rfloor +j}
\end{equation}
\end{problem}

\section{Recurrence Relations}  % Section 4
Let $s,t,u\colon\Natural\to\real$ be sequences satisfying the initial condition
\begin{equation}                % equation (4.1)
\label{eq41}
s(1)=t(1)=1,\quad u(1)=v(1)=0
\end{equation}
and the simultaneous recurrence relations (also called difference equations):
\begin{align}                % equation (4.2)
\label{eq42}
s(n+1)&=s(n)+v(n)\notag\\
t(n+1)&=t(n)+s(n)\\
u(n+1)&=u(n)+t(n)\notag\\
v(n+1)&=v(n)+u(n)\notag
\end{align}
We would like to find the solutions $s,t,u,v$ of \eqref{eq41} and \eqref{eq42}. To get an idea of these solutions, we give a table of their first few values. It is of interest to compare these values to $2^{n-2}$. Also, the fact that the numbers in \eqref{eq36} and the numbers in line $n=12$ agree is no coincidence.
\vskip 2pc

% Table 1
\begin{tabular}{|c|c|c|c|c|c|}
\hline
$n$&$s(n)$&$t(n)$&$u(n)$&$v(n)$&$2^{n-2}$\\
\hline
1&1&1&0&0&$1/2$\\
\hline
2&1&2&1&0&1\\
\hline
3&1&3&3&1&2\\
\hline
4&2&4&6&4&4\\
\hline
5&6&6&10&10&8\\
\hline
6&16&12&16&20&16\\
\hline
7&36&28&28&36&32\\
\hline
8&72&64&56&64&64\\
\hline
9&136&136&120&120&128\\
\hline
10&256&272&256&240&256\\
\hline
11&496&528&528&496&512\\
\hline
12&992&1024&1056&1024&1024\\
\hline
13&2016&2016&2080&2080&2048\\
\hline
14&4096&4032&4096&4160&4096\\
\hline
15&8256&8128&8128&8256&8192\\
\hline
\noalign{\smallskip}
\multicolumn{6}{c}%
{\textbf{Table 1}}\\
\end{tabular}
\vskip 2pc

The next result shows that we can replace the four simultaneous recurrence relations \eqref{eq42} by a single higher order recurrence relation.

\begin{thm}    % Theorem 4.1
\label{thm41}
The sequences, $s,t,u,v$ satisfy \eqref{eq41}, \eqref{eq42} if and only if they satisfy the following third order recurrence relation:
\begin{equation}                % equation (4.3)
\label{eq43}
w (n+3)=4w(n+2)-6w(n+1)+4w(n)
\end{equation}
together with the initial conditions
\begin{align}                % equation (4.4)
\label{eq44}
s(1)&=s(2)=s(3)=1;\ t(1)=1,t(2)=2, t(3)=3;\notag\\
u(1)&=0, u(2)=1, u(3)=3;\ v(1)=v(2)=0,v(3)=1
\end{align}
\end{thm}
\begin{proof}
Suppose $s,t,u,v$ satisfy \eqref{eq43}, \eqref{eq44}. It is easy to check that $s,t,u,v$ satisfy \eqref{eq41}, \eqref{eq42} for $n=1,2,3$. Proceeding by induction on $n$, suppose $s,t,u,v$ satisfy \eqref{eq42} for $1,2,\ldots ,n$, where $n\ge 3$. Then by \eqref{eq43} we have
\begin{align*}
v(n+1)&+s(n+1)\\
&=4v(n)-6v(n-1)+4v(n-2)+4s(n)-6s(n-1)+4s(n-2)\\
&=4s(n+1)-6s(n)+4s(n-1)=s(n+2)
\end{align*}
It follows by induction that the first equation in \eqref{eq42} holds and the other equations in \eqref{eq42} hold in a similar manner. Conversely, suppose $s,t,u,v$ satify \eqref{eq41}, \eqref{eq42}. It is straightforward to check that $s,t,u,v$ satisfy \eqref{eq44}. To show that $s$ satisfies \eqref{eq43} we have
\begin{align}                % equation (4.5)
\label{eq45}
4s(n+2)-6s&(n+1)+4s(n)\notag\\
&=4s(n+1)+4v(n+1)-6s(n)-6v(n)+4s(n)\notag\\
&=4s(n+1)+4v(n+1)-2s(n)-6v(n)\notag\\
&=4s(n)+4v(n)+4v(n)+4u(n)-2s(n)-6v(n)\notag\\
&=2s(n)+4u(n)+2v(n)
\end{align}
Applying \eqref{eq42} we have
\begin{align}                % equation (4.6)
\label{eq46}
t(n+1)+v(n+1)&=s(n)+t(n)+u(n)+v(n)\notag\\
   &=s(n+1)+u(n+1)
\end{align}
It follows from \eqref{eq42} and \eqref{eq46} that
\begin{align}                % equation (4.7)
\label{eq47}
s(n+3)&=s(n+2)+v(n+2)=s(n+1)+2v(n+1)+u(n+1)\notag\\
   &=s(n)+3v(n)+3u(n)+t(n)\notag\\
   &=s(n)+2v(n)+3u(n)+\sqbrac{t(n)+v(n)}\notag\\
   &=2s(n)+4u(n)+2v(n)(((
\end{align}
Comparing \eqref{eq45} and \eqref{eq47} shows that $s$ satisfies \eqref{eq43}. By symmetry $t,u,v$ satisfy \eqref{eq43}.
\end{proof}

\begin{problem}   % Problem 3
Solve \eqref{eq43} with initial conditions \eqref{eq44}.
\end{problem}

\section{Iterated Vector}  % Section 5
Let $1_n$ be the vector in $\real ^{2^n}$ all of whose components is $1$. Start with the vector $z(1)=(0,1)\in\real ^2$ and define the vectors $z(n)\in\real ^{2^n}$ recursively by
\begin{equation}                % equation (5.1)
\label{eq51}
z(n+1)=\paren{z(n),z(n)+1_n}
\end{equation}
For example,
\begin{align*}
z(2)=\paren{z(1),z(1)+1_1}=\paren{(0,1),(0,1)+(1,1)}=(0,1,1,2)\\
\intertext{and}
z(3)=\paren{z(2),z(2)+1_2}=(0,1,1,2,1,2,2,3)
\end{align*}
Now components of $z(n)$ have values in $\brac{0,1,\ldots ,n}$ and we are interested in the number of components $c_j(n)$ that equal $j{\pmod 4}$, $j=0,1,2,3$.

\begin{example} % Example 3
We have that $c_0(1)=c_1(1)=1$, $c_2(1)=c_3(1)=0$;
\begin{align*}
c_0(2)=1,\ c_1(2)=2,\ c_2(2)=1,\ c_3(2)=0\\
c_0(3)=1,\ c_1(3)=3,\ c_2(3)=3,\ c_3(3)=1
\end{align*}
and since
\begin{equation*}
z(4)=(0,1,1,2,1,2,2,3,1,2,2,3,2,3,3,4)
\end{equation*}
we get
\begin{equation*}
c_0(4)=2,\ c_1(4)=4,\ c_2(4)=6,\ c_3(4)=4
\end{equation*}
Compare these values with the first three lines of Table~1.\hfill\qedsymbol
\end{example}

For an application of this, let $z_j(n)$ be the $j$th component of $z(n)$. It is not hard to show that $z_j(n)$ is the number of $1$s in the $n$-bit binary representation of $j$, $j=0,1,\ldots ,2^n-1$. Then $c_j(n)$ is the number of integers that have $j$ $1$s $\pmod 4$ in their $n$-bit binary representation $j=0,1,\ldots ,2^n-1$.

\begin{example} % Example 4
For $n=1$, the $1$-bit representations are $0=0$, $1=1$ and $z_0(1)=0$, $z_1(1)=1$. For $n=2$, the $2$-bit representations are $0=00$, $1=01$, $2=10$, $3=11$ and $z_0(2)=0$, $z_1(2)=1$, $z_2(2)=1$, $z_3(2)=2$. For $n=3$, the $3$-bit representations are $0=000$, $1=001$, $2=010$, $3=011$, $4=100$, $5=101$, $6=110$, $7=111$ and $z_0(3)=0$, $z_1(3)=1$, $z_2(3)=1$, $z_3(3)=2$, $z_4(3)=1$, $z_5(3)=2$, $z_6(3)=2$, $z_7(3)=3$.\hfill\qedsymbol
\end{example}

\begin{problem}   % Problem 4
For $n\in\Natural$, find $c_j(n)$, $j=0,1,2,3$.
\end{problem}

\section{Relationships}  % Section 6
This section considers relationships between our four stated problems. First consider Problem~4. Examining \eqref{eq51} we see that 
\begin{equation}                % equation (6.1)
\label{eq61}
c_j(n+1)=c_j(n)+c_{j-1}(n)
\end{equation}
for $j=0,1,2,3$ where, by convention, $c_{-1}(n)=c_3(n)$. Also, $c_0(1)=1$, $c_1(1)=1$, $c_2(1)=c_3(1)=0$. We conclude that the sequences, $c_0,c_1,c_2,c_3$ satisfy the same initial conditions \eqref{eq41} and recurrence relations \eqref{eq42} as $s,t,u,v$ in Section~4. Since the solutions are unique, we conclude that
$c_0=s$, $c_1=t$, $c_2=u$, $c_3=v$. Moreover, by Theorem~\ref{thm41} these solutions coincide with the solutions of \eqref{eq43} with initial conditions
\eqref{eq44}. Hence, Problems~3 and 4 are equivalent.

Next consider Problem~2 and let $b_j(n)$ be the expression given in \eqref{eq37}. We now apply the combinatorial identity 
\begin{equation}                % equation (6.2)
\label{eq62}
\binom{n+1}{j}=\binom{n}{j-1}+\binom{n}{j}\quad 1\le j\le n
\end{equation}
We can prove \eqref{eq62} analytically or by the following combinatorial argument. Consider a set $S$ with $\ab{S}=n+1$ and let $s\in S$. There are
$\binom{n}{j-1}$ subsets of $S$ with cardinality $j$ that contain $s$ and $\binom{n}{j}$ subsets of $S$ with cardinality $j$ not containing $s$. These sum to the total cardinality $\binom{n+1}{j}$ of subsets of $S$. Using the usual convention $\binom{n}{j}=0$ for $j<0$, by \eqref{eq62} we have that
\begin{equation}                % equation (6.3)
\label{eq63}
b_j(n+1)=b_j(n)+b_{j-1}(n)
\end{equation}
for $j=0,1,2,3$ where again $b_{-1}(n)=b_3(n)$. Also, $b_0(1)=b_1(1)=1$, $b_2(1)=b_3(1)=0$. Comparing \eqref{eq61} and \eqref{eq63} we concluded that the three problems, 2,3,4, are equivalent.

We now treat Problem~1 which was the original motivation for this study. Each $n$-path in $\Omega _n$ can be thought of as a binary representation of a nonnegative integer. For example, in $\Omega _2$, $000,001,010,011$ correspond to $0,1,2,3$. We can thus think of $\Omega _n$ as $\Omega _n=\brac{0,1,\ldots ,2^n-1}$. For $j\in\Omega _n$, let $s_j(n)$ be the number of site or bit switches (or bit flips) in the binary representation of $j$. For instance, in $\Omega _2$, $s_0(2)=0$, $s_1(2)=1$, $s_2(2)=2$, $s_3(2)=1$. It follows from \eqref{eq22} that the amplitude of $j\in\Omega _n$ is 
\begin{equation*}
a(j)=\tfrac{1}{2^{n/2}}\,i^{s_j(n)}
\end{equation*}
For nonnegative integers $j,k$ define their parity $p_{jk}$ by 
\begin{equation*}
p_{jk}=
\begin{cases}
1&\hbox{if $j$ and $k$ are both even or both odd}\\
0&\hbox{otherwise}
\end{cases}
\end{equation*}
Applying \eqref{eq32} the decoherence of $j,k\in\Omega _n$ becomes
\begin{equation}                % equation (6.4)
\label{eq64}
D^n(j,k)=\tfrac{1}{2^n}\,i^{\sqbrac{s_j(n)-s_k (n)}}p_{jk}
\end{equation}
We call the matrix $D^n$ with entries $D^n(j,k)$ the \textit{time}-$n$ \textit{decoherence matrix}, $j,k=0,1,\ldots ,2^n-1$. For example, we have
\begin{equation*}
D^2=\tfrac{1}{4}
\left[\begin{matrix}\noalign{\smallskip}1&0&-1&0\\0&1&0&1\\-1&0&1&0\\0&1&0&1\\\noalign{\smallskip}\end{matrix}\right]
\end{equation*}

The reader should check using \eqref{eq64} that the entries of $D^n$ are $0$ or $\pm 1/2^n$ and if $p_{jk}=1$ then $D_{jk}^n=1/2^n$ when $s_j(n)=s_k(n){\pmod 4}$ and $D_{jk}^n=-1/2^n$ when $s_j(n)\ne s_k(n){\pmod 4}$. Moreover,
\begin{equation}                % equation (6.5)
\label{eq65}
\sum _{j,k=0}^{2^n-1}D_{jk}^n=1
\end{equation}
For instance, we immediately see that $D^2$ has these properties. It follows from \eqref{eq65} that $\mu _n(\Omega _n)=1$. It can be shown that $D^n$ is a positive semi-definite matrix and this implies that $\mu _n(A)\ge 0$ for every $A\subseteq\Omega _n$.

Letting $A'_n\subseteq\Omega _n$ be the set of Problem~1, applying \eqref{eq65} we have
\begin{equation}                % equation (6.6)
\label{eq66}
\mu _n(A'_n)=\sum _{j,k=1}^{2^n-1}D_{jk}^n=1-\sum _{j=1}^{2^n-1}D_{j0}^n-\sum _{k=0}^{2^n-1}D_{0k}^n
\end{equation}
Since $D^n$ is symmetric, writing \eqref{eq66} in terms of the values of $D_{j0}^n$ gives
\begin{align}                % equation (6.7)
\label{eq67}
\mu _n(A'_n)&=1-\tfrac{1}{2^n}-\tfrac{1}{2^n-1}\sum _{j=2}^{2^n-2}\brac{i^{s_j(n)}\colon j\hbox{ even}}\notag\\
   &=1-\tfrac{1}{2^n}-\tfrac{1}{2^{n-1}}\sum _{j=1}^{2^{n-1}-1}i^{s_{2j}(n)}\notag\\
   &=1+\tfrac{1}{2^n}-\tfrac{1}{2^{n-1}}\sum _{j=0}^{2^{n-1}-1}i^{s_{2j}(n)}
\end{align}

\begin{example} % Example 5
To compute $\mu _4(A'_4)$ we apply \eqref{eq67} to obtain
\begin{align*}
\mu _4(A'_4)&=1+\tfrac{1}{16}-\tfrac{1}{8}\sum _{j=0}^7i^{s_{2j}(4)}=\tfrac{17}{16}-\tfrac{1}{8}(1-1-1-1-1+1-1-1)\\
&=\tfrac{25}{16}\hskip 22pc\qed
\end{align*}
\end{example}

By \eqref{eq67} we might obtain a closed form formula for $\mu _n(A'_n)$ if we can get more information about $s_j(n)$. But $s_j(n)$ reminds us of $z_j(n)$ in Section~5. In fact, a physicist could call $z_j(n)$ and $s_j(n)$ conjugate or complementary variables because $z_j(n)$ measures the number of $1$s which can be used to compute the average position and $s_j(n)$ measures the number of site changes which can be used to compute the average momentum of the particle.

The next lemma will be useful in finding a relationship between $s_j(n)$ and $z_j(n)$.

\begin{lem}    % Lemma 6.1
\label{lem61}
$s_{2^{n+1}-1-j}(n+1)=s_j(n)+1, n\in\Natural, j=0,1,\ldots ,2^n-1$
\end{lem}
\begin{proof}
Let $j\in\Omega _n=\brac{0,1,\ldots ,2^n-1}$ and for $a\in\brac{0,1}$, let $a'=a+1 \pmod 2$. If $j$ has binary representation $j=a_0a_1\cdots a_n$, $a_0=0$, $a_k\in\brac{0,1}$, $k=1,\ldots ,n$, since
\begin{equation*}
a_0a_1\cdots a_n+a'_0a'_1\cdots a'_n=2^{n+1}-1
\end{equation*}
we have that
\begin{equation*}
(2^{n+1}-1)-j=0a'_0a'_1\cdots a'_n\in\Omega _{n+1}
\end{equation*}
Suppose $s_j(n)=k$ so $a_0a_1\cdots a_n$ has $k$ bit switches. These bit switches are in one-to-one correspondence with the bit switches in $a'_0a'_1\cdots a'_n$. Since $a'_0=1$, $0a'_0a'_1\cdots a'_n$ has one more bit switch so
\begin{equation*}
s_{2^n-1-j}(n+1)=k+1\qedhere
\end{equation*}
\end{proof}

To compare $s_j(n)$ with $z_j(n)$, let $y(n)$ be the vector
\begin{equation*}
y(n)=\paren{s_0(n),s_1(n),\ldots ,s_{2^n-1}(n)}
\end{equation*}
in $\real ^{2^n}$. Letting $T_n\to\real ^{2^n}\to\real ^{2^n}$ be the operator
\begin{equation*}
T(b_0,b_1,\ldots ,b_{2^n-1})=(b_{2^n-1},\ldots ,b_1,b_0)
\end{equation*}
it follows from Lemma~\ref{lem61} that $y(1)=(0,1)$ and
\begin{equation}                % equation (6.8)
\label{eq68}
y(n+1)=\paren{y(n),T_n\paren{y(n)+1_n}}
\end{equation}
We conclude that $s_j(n)$ has the same values as $z_j(n)$ only in a different order.

\begin{example} % Example 6
By \eqref{eq68} we have
\begin{align*}
y(2)&=\paren{y(1),T_1\paren{\paren{y(1)+1_1}}}=\paren{(0,1),T_1\paren{(0,1)+(1,1)}}\\
   &=\paren{(0,1),(2,1)}=(0,1,2,1)\\
   y(3)&=\paren{y(2),T_2\paren{y(2)+1_2}}=(0,1,2,1,2,3,2,1)\\
   y(4)&=(0,1,2,1,2,3,2,1,2,3,4,3,2,3,2,1)
\end{align*}
Compare these with $z(n)$ in Section~5.\hfill\qedsymbol
\end{example}

Since all paths begin with $0$ and paths of the form $2j$ end with $0$, $s_{2j}(n)$ must be even. Hence, $s_{2j}(n)=0\pmod 4$ or $s_{2j}(n)=2\pmod 4$. If
$s_{2j}(n)=0\pmod 4$, then $i^{s_{2j}(n)}$ contributes a $1$ to the sum in \eqref{eq67} and if $s_{2j}(n)=2\pmod 4$, then $i^{s_{2j}(n)}$ contributes a $-1$ to the sum in \eqref{eq67}. Since $s_j(n)$ and $z_j(n)$ have the same range, the number of contributed $1$s is $c_0(n)$ and the number of contributed $-1$s is $c_2(n)$ of Section~5. We conclude that 

\begin{equation}                % equation (6.9)
\label{eq69}
\mu _n(A'_n)=1+\tfrac{1}{2^n}-\tfrac{1}{2^{n-1}}\sqbrac{c_0(n)-c_2(n)}
\end{equation}
In summary, a solution to one of the equivalent Problems 2,3,4 will produce a solution to Problem~1.

\section{Solutions}  % Section 7
This section presents three quite different methods for solving our posed problems. As we shall see, each method has its advantages and disadvantages.

This first method is suggested by the techniques used in Section~3 to obtain \eqref{eq33}. Applying the binomial theorem gives 
\begin{align}               % equation (7.1)
\label{eq71}
(1+i)^n=\sum _{j=0}^n\binom{n}{j}i^j
\intertext{and}               % equation (7.2)
\label{eq72}
(1-i)^n=\sum _{j-0}^n\binom{n}{j}(-1)^ji^j
\end{align}
Adding \eqref{eq71} and \eqref{eq72} we have
\begin{equation}                % equation (7.3)
\label{eq73}
(1+i)^n+(1-i)^n=2\sum _{j=0}^{2\lfloor\frac{n}{2}\rfloor}\brac{\binom{n}{j}i^j\colon j\hbox{ even}}=2\sum _{k=0}^{\lfloor\frac{n}{2}\rfloor}(-1)^k\binom{n}{2k}
\end{equation}
Subtracting \eqref{eq72} from \eqref{eq71} we have
\begin{align}                % equation (7.4)
\label{eq74}
(1+i)^n+(1-i)^n&=2\sum _{j=1}^{2\lfloor\frac{n-1}{2}\rfloor +1}\brac{\binom{n}{j}i^j\colon j\hbox{ odd}}\notag\\
  &=2\sum _{k=1}^{\lfloor\frac{n-1}{2}\rfloor +1}(-1)^k\binom{n}{2k-1}
\end{align}
Applying \eqref{eq73} gives
\begin{align*}
\cos\frac{n\pi}{4}&=\frac{e^{in\pi/4}+e^{-in\pi/4}}{2}=\frac{1}{2}\sqbrac{(e^{i\pi/4})^n+(e^{-i\pi/4})^n}\\
  &=2^{-\frac{n}{2}-1}\sqbrac{(1+i)^n+(1-i)^n}=2^{-n/2}\sum _{k=0}^{\lfloor\frac{n}{2}\rfloor}(-1)^k\binom{n}{2k}
\end{align*}
Hence,
\begin{equation}                % equation (7.5)
\label{eq75}
\sum _{k=0}^{\lfloor\frac{n}{2}\rfloor}(-1)^k\binom{n}{2k}=2^{n/2}\cos\frac{n\pi}{4}
\end{equation}
In a similar way, applying \eqref{eq74} gives
\begin{equation}                % equation (7.6)
\label{eq76}
\sum _{k=0}^{\lfloor\frac{n-1}{2}\rfloor}(-1)^k\binom{n}{2k+1}=2^{n/2}\sin\frac{n\pi}{4}
\end{equation}
Adding \eqref{eq35} and \eqref{eq75} we obtain
\begin{equation}                % equation (7.7)
\label{eq77}
\binom{n}{0}+\binom{n}{4}+\binom{n}{8}+\cdots +\binom{n}{4\lfloor\frac{n}{4}\rfloor}=2^{n-2}+2^{\frac{n}{2}-1}\cos\frac{n\pi}{4}
\end{equation}
Subtracting \eqref{eq75} from \eqref{eq36} we obtain
\begin{equation}                % equation (7.8)
\label{eq78}
\binom{n}{2}+\binom{n}{6}+\binom{n}{10}+\cdots +\binom{n}{4\lfloor\frac{n-2}{4}\rfloor +2}=2^{n-2}-2^{\frac{n}{2}-1}\cos\frac{n\pi}{4}
\end{equation}
Similarly, adding and subtracting \eqref{eq35} and \eqref{eq76} gives
\begin{align}               % equation (7.9)
\label{eq79}
\binom{n}{1}+\binom{n}{5}+\binom{n}{9}+\cdots +\binom{n}{4\lfloor\frac{n-1}{4}\rfloor +1}&=2^{n-2}+2^{\frac{n}{2}-1}\sin\frac{n\pi}{4}\\
\label{eq710}              % equation (7.10)
\binom{n}{3}+\binom{n}{7}+\binom{n}{11}+\cdots +\binom{n}{4\lfloor\frac{n-3}{4}\rfloor +3}&=2^{n-2}-2^{\frac{n}{2}-1}\sin\frac{n\pi}{4}
\end{align}
Combining \eqref{eq77}--\eqref{eq710} gives the following solution to Problem~2.

\begin{thm}    % Theorem 7.1
\label{thm71}
For $n=1,2,\ldots $, and $j=0,1,2,3$ we have
\begin{equation*}
\binom{n}{j}+\binom{n}{n+4}+\binom{n}{n+8}+\cdots +\binom{n}{4\lfloor\frac{n-j}{4}\rfloor +j}=2^{n-2}+2^{\frac{n}{2}-1}\cos\frac{(n-2j)\pi}{4}
\end{equation*}
\end{thm}
This last method for solving Problem~2 had the advantage of being relatively short and concise, However, it was specialized and relied heavily on the binomial theorem. Our second method is suggested by Problem~3. Although longer than our previous method, it is more instructive because it illustrates the general technique of solving recurrence relations with initial conditions.

To solve the recurrence relation \eqref{eq43} we try a trial solution $w(n)=2^{\lambda n}$ where $\lambda$ is an unknown parameter to be determined. Substituting into \eqref{eq43} gives
\begin{equation}                % equation (7.11)
\label{eq711}
2^{3\lambda}2^{\lambda n}=4\,2^{2\lambda}2^{\lambda n}-6\,2^{\lambda}2^{\lambda n}+4\,2^{\lambda n}
\end{equation}
Cancelling the $2^{\lambda n}$ and letting $x=2^\lambda$ we obtain
\begin{equation}                % equation (7.12)
\label{eq712}
x^3-4x^2+6x-4=0
\end{equation}
We can factor \eqref{eq712} to get
\begin{equation}                % equation (7.13)
\label{eq713}
(x-2)(x^2-2x+2)=0
\end{equation}
The solutions of \eqref{eq73} are $x_1=2$, $x_2=1+i$, $x_3=1-i$. The corresponding values of $\lambda$ satisfy $2^{\lambda _1}=2$, $2^{\lambda _2}=1+i$, $2^{\lambda _3}=1-i$. We conclude that $\lambda _1=1$, $\lambda _2\ln 2=\ln (1+i)$ and $\lambda _3\ln 2=\ln (1-i)$. Hence, $\lambda _1=1$ and
\begin{align*}
\lambda _2&=\tfrac{1}{2}+\tfrac{i\pi}{4\ln 2}\\
\lambda _3&=\tfrac{1}{2}-\tfrac{i\pi}{4\ln 2}
\end{align*}
From the theory of recurrence relations (which is similar to differential equations) we conclude that the general solution of \eqref{eq43} has the form
\begin{align}               % equation (7.14)
\label{eq714}
w(n)&=c_12^{\lambda _1n}+c_22^{\lambda _2n}+c_32^{\lambda _3n}\notag\\
&=c_12^n+c_22^{n/2}e^{i\pi n/4}+c_32^{n/2}e^{-i\pi n/4}
\end{align}
where $c_1,c_2,c_3$ are arbitrary complex constants.

To find the constants $c_1,c_2,c_3$, we apply the initial conditions \eqref{eq44}. For $s(1)=s(2)=s(3)=1$ we have from \eqref{eq74} that 
\begin{align}               % equation (7.15)
\label{eq715}
2c_1+(1+i)c_2+(1-i)c_3&=1\notag\\
4c_1+2ic_2-2ic_3&=1\notag\\
8c_1-2(1-i)c_2-2(1+i)c_3&=1
\end{align}
Ordinarily, it would be a little tedious to solve the simultaneous equations \eqref{eq75} but we see by inspection that $c_1=c_2=c_3=1/4$. Hence, our particular solution becomes 
\begin{align*}
s(n)&=2^{n-2}+2^{\frac{n}{2}-2}e^{i\pi n/4}+2^{\frac{n}{2}-2}e^{-i\pi n/4}\\
   &=2^{n-2}+2^{\frac{n}{2}-1}\cos\frac{n\pi}{4}
\end{align*}
For $t(1)=1$, $t(2)=2$, $t(3)=3$ 23 have from \eqref{eq714} that
\begin{align*}
2c_1+(1+i)c_2+(1-i)c_3&=1\\
4c_1+2ic_2-2ic_3&=2\\
8c_1-2(1-i)c_2-2(1+i)c_3&=3
\end{align*}
Again, by inspection we have that $c_1=1/4$, $c_2=-i/4$, $c_3=i/4$. Hence, the particular solution becomes
\begin{equation*}
t(n)=2^{n-2}+2^{\frac{n}{2}-1}\sin\frac{n\pi}{4}
\end{equation*}
In a similar way, we obtain
\begin{align*}
u(n)&=2^{n-2}-2^{\frac{n}{2}-1}\cos\frac{n\pi}{4}\\
v(n)&=2^{n-2}-2^{\frac{n}{2}}\sin\frac{n\pi}{4}
\end{align*}
Of course, these solutions agree with the solutions \eqref{eq77}, \eqref{eq79}, \eqref{eq78}, \eqref{eq710}, respectively.

Our last method of solutions is mathematical induction. This is the shortest method but is the least instructive and has the disadvantage that we have to know the answer beforehand. We frame this in the form of the recurrence relations.

\begin{thm}    % Theorem 7.2
\label{thm72}
Let $v_j(n)$, $n=1,2,\ldots\,$, $j=0,1,2,3$ be four sequences satisfying $v_j(n+1)=v_j(n)+v_{j-1}(n)$ and the initial conditions $v_0(1)=v_1(1)=1$,
$v_2(1)=v_3(1)=0$ where $v_{-1}(n)=v_3(n)$. The the solutions are
\begin{equation}                % equation (7.16)
\label{eq716}
v_j(n)=2^{n-2}+2^{\frac{n}{2}-1}\cos\frac{(n-2j)\pi}{4}
\end{equation}
\end{thm}
\begin{proof}
We use mathematical induction on $n$. By the initial conditions, \eqref{eq716} holds for $n=1$, $j=0,1,2,3$. Suppose \eqref{eq716} holds for $n$ and $j=0,1,2,3$. We then have
\begin{align*}
v_j(n+1)&=v_j(n)+v_{j-1}(n)\\
   &=2^{n-1}+2^{\frac{n}{2}-1}\sqbrac{\cos (n-2j)\pi/4+\cos\paren{n-2(j-1)}\pi/4}\\
   &=2^{n-1}+2^{\frac{n}{2}-1}\sqbrac{\cos (n-2j)\pi/4-\sin(n-2j)\pi/4}\\
   &=2^{n-1}+2^{\frac{n}{2}-1}2^{1/2}\cos\sqbrac{(n-2j)\pi/4+\pi/4}\\
   &2^{(n+1)-2}+2^{\frac{(n+1)}{2}-1}\cos\sqbrac{\paren{(n+1)-2j}\pi/4}
\end{align*}
The result now follows by induction.
\end{proof}

Finally, to solve Problem~1, by Theorem~\ref{thm72} we have
\begin{align*}
c_0(n)&=v_0(n)=2^{n-2}+2^{\frac{n}{2}-1}\cos\frac{n\pi}{4}\\
c_2(n)&=v_2(n)=2^{n-2}+2^{\frac{n}{2}-1}\cos\frac{(n-4)\pi}{4}=2^{n-2}-2^{\frac{n}{2}-1}\cos\frac{n\pi}{4}
\end{align*}
Applying \eqref{eq69} gives
\begin{equation*}
\mu (A'_n)=1+\frac{1}{2^n}-\frac{1}{2^{n-1}}2^{\frac{n}{2}}\cos\frac{n\pi}{4}=1+\frac{1}{2^n}-\frac{1}{2^{\frac{n}{2}-1}}\cos\frac{n\pi}{4}
\end{equation*}
which solves Problem~1. We conclude that $\lim\mu (A'_n)=1$.

\end{document}